\documentclass[11pt]{amsart}
\usepackage{amsmath,amsthm,amscd,amsfonts,amssymb}
\usepackage[all]{xy}
\usepackage[ansinew]{inputenc} 

\newtheorem{thm}{Theorem}[section]  

\newtheorem{lema}[thm]{Lemma}
\newtheorem{prop}[thm]{Proposition}
\theoremstyle{definition}
\newtheorem{defi}[thm]{Definition}

\theoremstyle{remark}

\numberwithin{equation}{section}

\begin{document}

\title[Klein-Gordon equation from
Maxwell-Lorentz dynamics]
 {Klein-Gordon equation from
 Maxwell-Lorentz dynamics}

\author{R. J.  Alonso-Blanco}

\address{Departamento de Matem\'{a}ticas,
  Universidad de Salamanca, Plaza de la Merced 1-4, E-37008 Salamanca,
  Spain.}
\email{ricardo@usal.es}

\begin{abstract}
We consider Maxwell-Lorentz dynamics: that is to say, Newton's law
under the action of a Lorentz's force which obeys the Maxwell
equations. A natural class of solutions are those given by the
Lagrangian submanifolds of the phase space when it is endowed with
the symplectic structure modified by the electromagnetic field. We
have found that the existence of this type of solution leads us
directly to the Klein-Gordon equation as a compatibility
condition. Therefore, surprisingly, quite natural assumptions on
the classical theory involve a quantum condition without any
process of limit. This result could be a partial response to the
inquiries of Dirac in \cite{Dirac1}.
\end{abstract}

\maketitle

\centerline{\today}

\section*{Introduction}

Since the beginning of Quantum Mechanics, physicists have linked
the new mechanics to the classical one. This is bound because it
is not possible to move forward in the vacuum and by the success
of the Newtonian and  Maxwellian theories. For example,
Schr\"{o}dinger \cite{Schrodinger} justifies his famous wave
equation by appealing to the classical Hamilton-Jacobi equation.
In this paper we will prove that the geometry of the Newton
equation (as exposed in \cite{MecanicaMunoz}) leads naturally to
the (generalized) Klein-Gordon equation. More specifically, we
will see that the Klein-Gordon equation is a necessary condition
for the existence of charge-current fields obeying the Lorentz
law, if such fields correspond to Lagrangian submanifolds (of the
symplectic structure of the phase space). The result
holds for arbitrary pseudo-metrics (on arbitrary dimensional
manifolds) included, as particular cases, Riemann and Minkowski
metrics.

We want to emphasize the possible relationship of this work with
the research of Dirac in \cite{Dirac1}(see also \cite{Parrot},
pp.190-3, for an interesting and modern treatment and
\cite{Dirac2,Dirac3} where Dirac himself extends the formalism in
order to include more general physical situations). In the above
cited paper the author re-examines the classical theory of
electrons: he first imposes the constancy of  $\|A\|^2$ as a
subsidiary condition on the ordinary action principle for the
electromagnetic field $F=dA$  (``{\sl the simplest relativistic
way of destroying the gauge transformations}'', in his words).
Then, a further analysis
 lead him to consider the vector field associated with the very $A$ as a possible motion of charges. On our part, we consider a natural
 kind of solutions of the Lorentz law (Lagrangian submanifolds). It turns out that the motions of charges thus described, consist of certain potentials
 $A$ which  are, necessarily, of constant length. In this way, it seems that we recover the kind of description that Dirac was
 looking for.

The structure of the paper is as follows. In the first section we
exhibit the formulation of Newton's second law as given in
\cite{MecanicaMunoz}, where its relationship with the Lagrangian
and Hamiltonian formulations becomes transparent. In the second
section, we describe the Newton equation when the force is given
by the Lorentz law and explain the equivalence with the no-force
case  at the cost of modifying the underlying symplectic
structure. In Section~3, we characterize velocity fields $u$ of
flows comprised of particles obeying Newton's law: \emph{field
solutions}. Finally, we derive the Klein-Gordon equation for field
solutions: in Section 4 for the free case and in Section 5 for the
charge-current tensor given by the Maxwell equations.

\vskip 2cm

\section{Newton's equation}
\bigskip

 Let $M$ be an smooth manifold endowed with a non-degenerate metric
 $T_2$ whose expression in local coordinates $q^\mu$ is
$$T_2=g_{\mu\nu}dq^\mu dq^\nu,\quad \textrm{det}\, (g_{\mu\nu})\ne 0.$$

The metric $T_2$ allows us to translate the canonical Liouville
1-form from the cotangent bundle $T^*M$ to the tangent bundle
$TM$. Let us denote by $\theta$ this translated form. Locally,
$$\theta=g_{\mu\nu}\dot q^\mu dq^\nu,$$
where, as usual, dotted functions mean $\dot f(v_x):=v_x(f)$ for
all tangent vector $v_x\in T_xM$, $x\in M$ and
$f\in\mathcal{C}^\infty(M)$. In this way, 2-form $\omega:=d\theta$
defines a symplectic structure on $TM$.

Let us denote by  $\pi\colon TM\to M$ the natural projection
$\pi(v_x)=x$. We will say that an 1-form $\alpha$ defined on the
manifold $TM$ is {\sl{horizontal}} if it kills all of the vertical
 (with respect to $\pi$) tangent vectors. In local coordinates,
$$\alpha=\alpha_\mu dq^\mu,\quad \alpha_\mu=\alpha_\mu(q,\dot q).$$

 A {\sl second order differential equation} is, by definition, a
vector field $D$ on the manifold $TM$ such that
$$\pi_*D_{v_x}=v_x,\quad \forall v_x\in TM.$$
This means that $D f=\dot f$ and then, locally,
$$D=\dot q^\mu\frac{\partial}{\partial {q^\mu}}+f^\mu(q,\dot
q)\frac{\partial}{\partial \dot q^\mu},$$ for suitable functions
$f^\mu$.

\medskip

Now, we will establish the Newton's second law ``$F=ma$'' in a very
convenient fashion.

\begin{defi}
The \sl{Newton's equation} associated with a horizontal 1-form
$\alpha$ (``work form''), is the only second order differential
equation $D$ such that
\begin{equation}\label{Newton}
D\lrcorner\,\omega+dT+\alpha=0
\end{equation}
where $T$ denotes the ``kinetic energy'' function defined by
$$T(v_x):=\frac 12 T_2(v_x,v_x),\quad \forall v_x\in TM,$$
(so that, locally, $T=(1/2)g_{\mu\nu}\dot q^\mu\dot q^\nu$).

\end{defi}

The reciprocal determination of  $D$ and $\alpha$ by means of
(\ref{Newton}) can be viewed in \cite{MecanicaMunoz}, from where
we have taken it: each ``force'' $\alpha$ determines an
``acceleration'' $D$ and vice versa. This proves the correctness
of the above definition. When $\alpha=0$ (no force is acting on
the system), equation (\ref{Newton}) is the equation of the
geodesic curves with respect to the metric $T_2$.

\bigskip

\section{Lorentz's force}

An electromagnetic field is defined as a closed 2-form $F$ on $M$,
$$F=F_{\mu\nu}dq^\mu\wedge dq^\nu,$$
and a potential for $F$ is an 1-form $A=A_\mu dq^\mu$ such that
 $F=dA$; that is to say,
 $$ F_{\mu\nu}=\frac 12\left(\frac{\partial A_\nu}{\partial q^\mu}-\frac{\partial A_\mu}{\partial
q^\nu}\right).$$

\begin{defi} The Lorentz force associated with $F$ is given by the 1-form $\alpha$ defined by
$$\alpha_{v_x}=v_x\lrcorner F,\quad \forall v_x\in TM.$$
(which depends on the velocities); locally,
$$\alpha=2F_{\mu\nu}\dot q^\mu\,dq^\nu.$$
The associated  Newton's equation (\ref{Newton}) now is
\begin{equation}\label{Newtonpotencial}
D\lrcorner(\omega+F)+dT=0\ .
\end{equation}
\end{defi}

\noindent Let us introduce the notation
$$\omega_F:=\omega+F$$
so that $\omega_F=d(\theta+A)$ if $F=dA$  and
 equation (\ref{Newtonpotencial}) can be rewritten as
\begin{equation}\label{Newtonpotencial2}
D\lrcorner\,\omega_F+dT=0\ ,
\end{equation}
which is interpreted as follows: $D$ is the Hamiltonian field
corresponding to the function  $T$, the kinetic energy, for the
symplectic structure $(TM,\omega_F)$.
\bigskip

\section{Field-solutions of Newton's equation}

We will deal now, not with the solution curves of Newton's
equation (which is of second order) taken individually, but rather
of vector fields whose integral curves are among the former
ones.\medskip

Let us observe that each parameterized curve on $M$ can be
naturally prolonged to $TM$: it is sufficient to consider the
velocity vector at each point.

A vector field $u\colon M\to TM$ is said to be an {\sl
intermediate integral} of a second order equation
 $D\colon TM\to TTM$ if and only if
 the prolongation to $TM$ of each integral curve of $u$ is also an integral curve of $D$.

The prolongation to $TM$ of integral curves of $u$, completely
fills the submanifold $Im(u)$, image of $u\colon M\to TM$. In
fact, $u_*u$ is the vector field on $Im(u)$ whose integral curves
are the prolongation of integral curves of $u$ (we denote by $u_*$
the tangent map associated with map $u$). As a consequence, the
condition
 for $u$ being and intermediate integral of $D$ is
\begin{equation}\label{intermedia}
   D_{u_x}=u_*u_x,\quad \forall x\in M. 
\end{equation}

On the other hand, it follows from the very definition of
Liouville 1-form that for an arbitrary vector field $u$, the
restriction of $\theta$ to the image of map $u$ is
\begin{equation}
u^*\theta=u\lrcorner T_2.
\end{equation}
 ($u^*$ denotes the pull-back induced by $u$).

 The intermediate integrals $u$ for Newton's equation
(\ref{Newton}) will be called {\sl field-solutions}. In this way,
and summing up the above considerations, we get the ``only if''
part of the following characterization.
\begin{prop}
A vector field $u$ is a field-solution of Newton's equation
(\ref{Newton}) if and only if
\begin{equation}\label{intermediaNewton}
u\lrcorner\,  d(u\lrcorner T_2)+dT(u)+u^*\alpha=0,
\end{equation}
where $T(u):=u^*T$.
\end{prop}
\begin{proof} 
The ``if'' part is due to $D_{u_x}-u_*u_x$ being vertical and then
the vanishing of $u^*((D_{u_x}-u_*u_x)\lrcorner\, \omega)$ implies
$D_{u_x}-u_*u_x=0$.
\end{proof}
\bigskip

\section{Field-solutions for the Lorentz force}

From now on, we will follow and extend a convention from
\cite{Parrot}: With each $(p,q)$ tensor field  $S$ on $M$, we will
denote $S^\star$ the $(q,p)$ tensor that is obtained when the
covariant components are converted into contravariant components
and vice versa, by means of the metric $T_2$. For instance, if $u$
is a tangent vector field on $M$,  $u^\star$ will be the 1-form
 $$u^\star:=u\lrcorner T_2,$$
or, in local coordinates,
$$u=u^\mu\frac{\partial}{\partial q^\mu},\quad u^\star=u_\mu dq^\mu,$$
where $u_\mu:=g_{\mu\nu}u^\nu$, according to the usual procedure
for raising and lowering indexes.

  Newton's equation
(\ref{intermediaNewton}) for a field $u$ under the Lorentz force
induced by $F=dA$  is
\begin{equation}\label{Lorentz2}
u\lrcorner\, d(u^\star+\,A)+\,dT(u) =0.
\end{equation}
In such a case, kinetic energy $T(u)$ is a first integral of $u$.
\medskip

A particularly important class of field solutions $u$ of
(\ref{Lorentz2}) is given by the set of Lagrangian submanifolds
(with respect to $\omega_F$).
Explicitly, this is equivalent to consider those tangent fields
$u\colon M\to TM$ such that, locally,
\begin{equation}\label{lagrangian}
u^\star+A=df \end{equation}
 for a  local smooth function $f$ defined on
$M$. In this case, we derive $dT(u)=0$ and then,
\begin{lema}
 The Lagrangian solutions $u$ of equation (\ref{Lorentz2}) have
 constant kinetic energy
 \begin{equation}\label{HJ}
 T(u)=\frac 12 m^2,
 \end{equation}
  ($m$ can be real or
imaginary, according to the sign of the kinetic energy).
\end{lema}
\bigskip

\section{Klein-Gordon equation for conservative geodesic fields}

As a first case, we will derive the Klein-Gordon equation without
coupling; so, no electromagnetic field is present: $F=0$.

Let $\delta$ denote the divergence operator associated with the
(pseu\-do-rie\-man\-nian) metric $T_2$. On vector fields, $\delta$
is the standard divergence operator. For an 1-form $\sigma$, we
have $\delta\sigma=\delta(\sigma^\star)$. The Laplacian operator
on differential forms is given by
$$\Delta:=d\delta+\delta d,$$
which simplifies to $\Delta=\delta d$ on functions (0-forms). As
well known, in the case of the Minkowski metric, $\Delta$ symbol
is substituted by $\square$ and named D'Alambertian operator.

\begin{lema}\label{lemalaplaciana}
For each smooth function $\phi$ on $M$ it holds
 $$\Delta e^{i\phi}=e^{i\phi}\left(i\Delta \phi-T^2(d\phi,d\phi)\right).$$
\end{lema}
\begin{proof}
 By definition,
   $$\Delta e^{i\phi}=\delta d(e^{i\phi})=\delta(ie^{i\phi}\,d\phi)
                  =\delta (i e^{i\phi}(d\phi)^\star)
                  =(d\phi)^\star(i e^{i\phi})+i e^{i\phi}\delta (d\phi)^\star,
                  $$
 where we applied the identity
  $\delta(\lambda v)=v(\lambda)+\lambda\delta v$
 for the  function $\lambda=i e^{i\phi}$ and the vector field $v=(d\phi)^\star$.

 But, $$(d\phi)^\star(i e^{i\phi})=i^2 e^{i\phi}(d\phi)^\star(\phi)=-e^{i\phi} T^2(d\phi, d\phi)$$ and
 $\delta (d\phi)^\star=\delta d\phi=\Delta\phi$, which finishes the proof.
 \end{proof}

Now, let us consider the equation for the geodesic vector fields.
This corresponds to take equation (\ref{Lorentz2}) for $F=0$,
\begin{equation}\label{geodesic}
u\lrcorner\, du^\star+\,dT(u) =0.
\end{equation}

A class of solutions of (\ref{geodesic}) is given by those vector
fields $u$ such that $du^\star=0$. Equivalently, with respect to
the symplectic manifold $(TM,\omega)$, we are considering  the
local Lagrangian submanifolds $u^\star=df$ for some function
$f\in\mathcal{C}^\infty(M)$.
\medskip

In particular, $T(u)=\frac 12 T_2(u,u)=\frac 12 T_2(df,df)=\frac
12 m^2$.
\medskip

Moreover, let us assume that field $u$ is conservative:
$$\delta u=0.$$
Then $\Delta f=0$ and Lemma \ref{lemalaplaciana}, when applied on
$\phi=\frac 1\hbar f$, gives
 $$\Delta e^{i\frac f\hbar}
    =e^{i\frac f\hbar}\left(i\frac 1\hbar\,\Delta f-\frac 1{\hbar^2}\,T^2(df,df)\right)
    =-e^{i\frac f\hbar}\frac{m^2}{\hbar^2}.$$
 So that,
 \begin{prop}
 Conservative geodesic fields $u^\star=df$ hold Klein-Gordon
 equation
 $$\left(\Delta+\frac{m^2}{\hbar^2}\right)\psi=0,$$
 for $\psi:=e^{i\frac f\hbar}$ and $m^2=T_2(u,u)$.
 \end{prop}
\bigskip

\section{Klein-Gordon equation for a Maxwell field}

 The first Maxwell equation for electromagnetic field $F=dA$ is
 $$\delta F=J^\star,$$
 where $J$ is the so called charge-current vector field. In particular, field $J$ is
 conservative because $\delta J=\delta^2F=0$.

 Let us assume that the vector field $u=J$ gives a Lagrangian solution of  Newton's equation for the Lorentz
 force (\ref{Lorentz2}) so that, locally exists a function $f$ such that
 $$J^\star+A=u^\star+A=df.$$
 (This is equivalent to say that $-J^\star$ is a vector potential for $F$ as was proposed by Dirac in \cite{Dirac1}).
 In addition, let us choose a potential $A$ with the {\sl Lorentz gauge}
  $$\delta A=0.$$
 In this case,
 $$\Delta f=\delta u^\star+\delta A=\delta J+\delta A=0.$$
 Now, the same computation of the previous section, gives us
 $$\Delta e^{i\frac f\hbar}
    =e^{i\frac f\hbar}\left(i\frac 1\hbar\,\Delta f-\frac 1{\hbar^2}\,T^2(df,df)\right)
    =-\frac 1{\hbar^2}e^{i\frac f\hbar}\,T^2(df,df).$$
 On the other hand,
 $$m^2=T^2(u^\star,u^\star)=T^2(-A+df,-A+df)
      =\|A\|^2-2T^2(A,df)+T^2(df,df),$$
 so that, taking into account  $T^2(A,df)=A^\star(f)$ (= the
 derivative of $f$ along
 the vector field $A^\star$),
 $$T^2(df,df)=m^2-\|A\|^2+2A^\star(f).$$
  It follows that
 $$\Delta e^{i\frac f\hbar}
        =-\frac 1{\hbar^2}e^{i\frac
        f\hbar}\,\left(m^2-\|A\|^2+2A^\star(f)\right)$$
 or
 $$\left(\Delta - 2\frac i\hbar A^\star+ \frac 1{\hbar^2}\left(m^2-\|A\|^2\right)\right)e^{i\frac f\hbar}=0$$
 where we have applied that $A^\star(e^{i\frac f\hbar})=\frac
 i{\hbar}e^{i\frac f\hbar}A^\star(f)$.
 Summing up,
 \begin{thm}
 If the Newton equation (\ref{Lorentz2}) for the Lorentz force given by $F=dA$, $\delta A=0$, admits a Lagrangian
 solution
    $J^\star+A=df,$ where $J^\star:=\delta F$,
 then $T_2(J,J)$ is a constant, say $m^2$ (positive or negative), and the exponential $\psi=e^{i\frac f\hbar}$ satisfies  the Klein-Gordon
 equation
 \begin{equation}\label{KGconA}
 \left(\Delta - 2\frac i\hbar A^\star+ \frac 1{\hbar^2}\left(m^2-\|A\|^2\right)\right)\psi=0
 \end{equation}
 \end{thm}
 \medskip

 Because function  $f$ is supposed to be real, equation (\ref{KGconA}) for $\psi=e^{if/\hbar}$ implies that $\Delta f=0$
 and $T_2(u,u)=m^2$, for $u^\star:=-A+df$. As a consequence, $u$ holds Newton's equation (\ref{Lorentz2}) and $\delta u=0$ (and conversely).
 However, the condition of being $u=\delta F$ can not be assured.
 \bigskip

 \section*{Aknowledgements}
 This work had its origin in the sessions of the seminar on mathematical foundations of physics organized by Professor Jesús Muñoz Díaz,
   my thesis advisor and dear friend. In the course of these investigations, he generously shared with me many  developments, results,
   insights and
   key ideas. In fact, this article relies very strongly on that
   material.
 So the possible merits of these pages, if any, are to be attributed more to him than myself. Therefore, my deepest gratitude.
 \bigskip

\renewcommand*{\refname}{References}


\begin{thebibliography}{9}
\bibitem{Dirac1}  Dirac, P.A.M., A new classical theory of electrons , Proc. R. Soc. Lond. A, 1951 vol.~209, no.
1098, pp. 291--296.

\bibitem{Dirac2}  Dirac, P.A.M., A New Classical Theory of Electrons II, Proc. R. Soc. Lond. A, 1952 vol. 212, no. 1110, pp.
330--339.

 \bibitem{Dirac3} Dirac, P.A.M., A New Classical Theory of Electrons III, Proc. R. Soc. Lond. A, 1954 vol. 223, no. 1555. pp.
 438--445.

\bibitem{MecanicaMunoz} Mu\~{n}oz D\'{\i}az, J., The structure of
time and inertial forces in Lagrangian mechanics, Contemporary
Mathematics, vol. 549, 2011, pp. 65--94.


\bibitem{Parrot} Parrott, S., \emph{Relativistic Electrodynamics and Differential Geometry}, Springer-Verlag, Nueva York
(1987).



\bibitem{Schrodinger} Schrödinger, E., Quantisierung als Eigenwertproblem, 
Annalen der Physik, (4), 79, (1926), 361-376.

\end{thebibliography}
\end{document}